\documentclass[11pt]{article}

\usepackage{amsmath,amsthm,amssymb,epsf,eucal}
\usepackage{bm} 
\usepackage{graphicx,psfrag}
\usepackage{epstopdf}
\usepackage{verbatim}
\usepackage{multirow}
\usepackage{color}
\usepackage{setspace}
\usepackage{subfigure}
\newtheorem{claim}{Claim}

\newtheorem{lemma}{Lemma}

\newcommand*{\bs}{\boldsymbol}

\title{Local likelihood estimation of local parameters for nonstationary random fields.}
\author{Ethan B. Anderes and Michael L. Stein}

\begin{document}
\maketitle
\frenchspacing

\begin{abstract}
We develop a weighted local likelihood estimate  for the parameters that govern the local spatial dependency of a locally stationary random field. The advantage of this local likelihood estimate is that it smoothly downweights the influence of far away observations, works for irregular sampling locations,  and when designed appropriately, can trade bias and variance for reducing estimation error. This paper starts with an exposition of our technique on the problem of estimating an unknown positive function when multiplied by a stationary random field. This example gives concrete evidence of the benefits of our local likelihood as compared to na\"ive local likelihoods where the stationary model is assumed  throughout a neighborhood.  We then discuss the difficult problem of estimating a bandwidth parameter that controls the amount of influence from distant observations.  Finally we present  a simulation experiment for estimating the local smoothness of a local Mat\'ern random field when observing the field at random sampling locations in $[0,1]^2$. 
\end{abstract}

\section{Introduction}
Stationary random fields play a  fundamental role in both theoretical and applied spatial statistics. 
Unfortunately, stationarity is often violated when working with real data. The issue is more pressing with the recent data deluge and high resolution sensing where non-stationarity can be clearly visible. This presents a challenge for the spatial statistician who is  interested in estimating and modeling dependency structure in random fields. 

Even though stationarity is often criticized as being too simplistic, we believe  it is still important for the understanding and development of nonstationary models. The reason is that any type of statistical estimation requires some sort of replication to \lq\lq average over".  In spatial statistics the data often comprise of one realization of a random field. In this case, the assumption of stationarity provides a type of replication that makes statistical estimation possible. For nonstationary random fields, however, the lack of any assumption leads to a breakdown in statistical estimation due to the absence of replication. This problem can be mitigated by adding assumptions on the nonstationary random field like local stationarity, for example. The idea is that on small enough spatial scales, one hopes that the local dependency of the random field near a point is well approximated by some stationary random field. In this paper, we do not attempt a precise definition of local stationarity (see \cite{Dahlh:00} for a definition in the time series literature). Instead, we take it for granted that  such random fields exist (however one defines it) and enter into a discussion of how one might estimate the parameters of a local stationary approximation when observing a single realization of the nonstationary random field at dense (possibly uneven) observation locations.  To accomplish these goals we develop a local likelihood approach for  estimating local parameters. 

Before we continue we mention why the standard local likelihood techniques fail for estimating the local dependency of a random field. In current techniques (see \cite{fan:98}, for example), the independence structure of the data typically allow one to decompose the log-likelihood as a sum. Each summand depends on one data point, which can be down weighted as a function of some spatial covariate. In the random field case, however, there is no independence and therefore no such decomposition of the log-likelihood. A different technique that does work for random fields is to simply divide the observation locations into neighborhoods and fit a stationary random field model on each neighborhood. Typically two problems arise with this approach. First, the range of validity of a stationary approximation can be too small to contain enough data to estimate it. Second, it can produce non-smooth  local parameter estimates, which can be undesirable in many cases. 
 
  In this paper we present an exposition through computation, simulation and some theory of our version of local likelihood estimation. A  large portion of the paper is devoted to the discussion of different ways of constructing and estimating the weights used in our local likelihood that downweights the influence of distant observations of the random field.  We start in Section \ref{sec2} with our definition of a weighted local likelihood and then immediately apply it to the problem of variance modulation  in Section \ref{VM}.
  This example is convenient since  the local likelihood estimate has a closed form and one can derive the Bayes risk under a polynomial prior for $\sigma(x)$. 
In Section \ref{Band} we study the problem of estimating the bandwidth parameter $\lambda$ from data. Finally, in Section \ref{2dExample} we apply these techniques to the estimation of the local fractional index of a local Mat\'ern random field when observing one realization of the field at uneven observation locations in $\Bbb R^2$.  

\section{Weighted Local Likelihood}
\label{sec2}

Before we present our notion of local likelihood it will be advantageous to set some notation. We write a random field as $\{ Z(t)\colon t\in\Bbb R^d \}$ or just $Z$ when the domain of definition is clear from context.  We distinguish two types of random field models: {\it global} and {\it local}. The only real distinction is nonstationary versus stationary but the nomenclature is useful since we regard global models as the true nonstationary sampling distribution and local models as the local stationary approximations. We consider global models that are nonstationary random fields indexed by some nonparametric function $\theta(t)$ that takes spatial arguments $t\in\Bbb R^d$ and returns values in some $m$-dimensional  parameter space $\Theta\subset \Bbb R^m$.
We call the function $\theta(\cdot)$ the  {\it local parameter function} and denote resulting global model $G_{\theta(\cdot)}$. For each fixed  $t_0\in\Bbb R^d$, the parameter vector $\theta_0=\theta(t_0)$ determines a {\it local} random field model, denoted by $L_{\theta_0}$, which is generally stationary and models the stochastic behavior of $Z$ near some point $t_0$. 
Informally this means that  $\mathcal L_{L_{\theta_0}}\{Z(t_0+h)\colon |h|<\epsilon  \}\approx \mathcal L_{G_{\theta(\cdot)}}\{Z(t_0+h)\colon |h|<\epsilon  \}$ when $\epsilon$ is small ($\mathcal L_{L_{\theta_0}}$ denotes the law of the finite dimensional distributions under model $L_{\theta_0}$). For the remainder of the paper a local parameter function will be denoted by $\theta(\cdot)$, $\theta(t)$ or just $\theta$ and a particular value in the parameter space $\Theta\subset \Bbb R^m$ will be written  $\theta_0$.

For a concrete example, let $Z$ be a random field on $\Bbb R^d$ such that there exists a  fixed but unknown function  $\sigma(t)\colon \Bbb R^d\rightarrow \Bbb R^+$ and a known stationary random field $W$ such that  $Z(t)=\sigma(t) W(t)$. In this case,  $\sigma(t)$ is the local parameter function and $G_{\sigma(\cdot)}$ denotes the true distribution of $Z$. If   $\sigma(t_0+h)\approx\sigma_0$ when $|h|$  is small then the local model for $Z$ at $t_0$, denoted by $L_{\sigma_0}$, is just the law of  $\sigma_0 W(t)$. Presumably if one has enough data in a neighborhood of $t_0$ one could successfully estimate $\sigma(t_0)$ by fitting the model $L_{\sigma_0}$ to the data.  Notice there is an inherent bias-variance tradeoff when fitting $L_{\sigma_0}$ locally to data: increasing the size of the local neighborhood reduces the variability of the estimate but increases the bias due to the inaccuracy of the stationary approximation. It is with local likelihoods that we seek to balance these two competing terms by smoothly down weighting the dependence of far away observations.

To define the local likelihood estimation of  $\theta(\cdot)$ suppose we have $n$ observations $(t_1,z_1),\ldots,(t_n,z_n)$ of a single realization of a random field $Z$ (the spatial locations are $t_j\in\Bbb R^d$ and the responses are $z_j=Z(t_j)$).
For any given location $t\in \Bbb R^d$, let  $\mathcal  N_{t,k}$ denote the set of $k$ observations nearest to $t$.  Given  weights $w_1,\ldots, w_n$  (possibly depending on $t$, $t_k$ and a bandwidth parameter $\lambda$) we define the following weighted local likelihood:
\begin{equation}
\label{locLike}
 \mathcal{ W}_{\lambda}(\theta_0, t|\text{data})\triangleq\sum_{k=1}^n w_k \bigl[ \ell(\mathcal  N_{t,k}| L_{\theta_0})-\ell(\mathcal  N_{t,k-1}|L_{\theta_0}) \bigr]
 \end{equation}
 where $\ell(\mathcal  N_{t,k}|L_{\theta_0})$ is the log likelihood of the data in $\mathcal N_{t,k}$ under the model $L_{\theta_0}$ (note: we define $\ell(\mathcal  N_{t,0}|L_{\theta_0})$ to be $0$).
Now our local likelihood estimate of $\theta(t)$ is defined as
\[ \hat{\theta}_\lambda(t)\triangleq \text{arg}\max_{\theta_0 \in \Bbb R^m}  \mathcal W_\lambda(\theta_0, t|\text{data}). \]
We first remark that if the weights $w_k=1$ for all $k$, then the telescoping sum in (\ref{locLike}) collapses and one recovers the full likelihood for the local stationary model. Indeed $\mathcal W_\lambda$ orders and subsequently downweights  the incremental changes in the stationary likelihood when adding the observations one by one in order of their distance to $t$.
An important feature of the estimate $\hat \theta_\lambda$ is that, at least for Gaussian random fields, the computational cost of  $\mathcal W_\lambda(\theta_0, t|\text{data})$ is comparable to that of  $\ell(\theta|\text{data})$  by either {\it up dating} a Cholesky decomposition or  {\it down dating} an inverse covariance matrix. Moreover, the estimate $\hat \theta_\lambda(t)$ will typically be a \lq\lq smooth" function of $t$ (when the weights are smooth) even when there is no natural  additive or multiplicative structure on the parameter space $\Theta\subset\Bbb R^m$. Finally we remark that $\hat\theta_\lambda$ depends on  the bandwidth parameter $\lambda$ through the weights $w_k$ which  typically have the form $w_k\triangleq  K((t-t_k)/\lambda)$ for some smoothing kernel $K$.

\section{Variance modulation} \label{VM}
We start with a particularly simple example, already mentioned in the previous section:  estimating the local variance of a Gaussian random field. Consider the estimation of the function $\sigma(t)\colon \Bbb R \rightarrow \Bbb R^+$ when observing $\sigma(t)W(t)$ where $W$ a mean zero Gaussian process on $\Bbb R$ with known Mat\'ern parameters.
This example is useful since the local maximum  likelihood estimate has a closed form solution. We take advantage of this solution by deriving the Bayes risk under polynomial priors. The Bayes risk is then used to give concrete evidence of a bias and variance trade for reducing estimation error. We also use the Bayes risk to explore the relationship between 
higher order kernels, bias of the resulting estimates and optimal bandwidth.

To derive the closed form solution of the local likelihood estimate we first need some notation. Let $\Sigma_{k}$ denote the covariance matrix of the observations in the $k$ neighborhood of $t$, i.e. $\mathcal N_{t,k}$, and let $\bs z_{k}\triangleq (z_1,\ldots,z_k)^T$ denote the responses of the data in $\mathcal N_{t,k}$. 
Notice that 
\[ \ell( \mathcal  N_{t,k}|L_{\sigma_0})=-\frac{ \bs z_{k}^T \Sigma_{k}^{-1} \bs z_{k}}{2\sigma_0^2}-\frac{\log|\sigma_0^2\Sigma_{k}|}{2}-\frac{k}{2}\log (2\pi)  \]
and therefore
 \[ \mathcal W_\lambda(\sigma_0,t|\text{data})=-\log \sigma_0\sum_{k=1}^n w_k+\frac{1}{2\sigma_0^2}\sum_{k=1}^n w_k \Bigl[\bs z_{k-1}^T \Sigma_{k-1}^{-1} \bs z_{k-1} -\bs z_{k}^T \Sigma_{k}^{-1} \bs z_{k} \Bigr] +c. \]
  The maximum occurs at 
 \begin{equation}
 \label{softEst}
  \hat \sigma_\lambda^2\triangleq \frac{\sum_{k=1}^n w_k \Bigl[\bs z_{k}^T \Sigma_{k}^{-1} \bs z_{k}-\bs z_{k-1}^T \Sigma_{k-1}^{-1} \bs z_{k-1}  \Bigr] }{\sum_{k=1}^n w_k}
  \end{equation}
 which is the local maximum weighted likelihood estimate of $\sigma^2$ at $t$. 
 
 \begin{figure}[t]
\centering
{\includegraphics[height=1.8in]{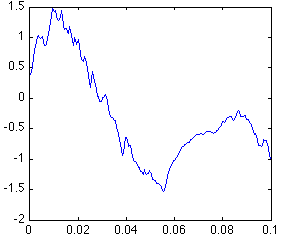}}
{\includegraphics[height=1.8in]{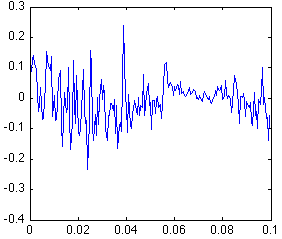}}
{\includegraphics[height=1.8in]{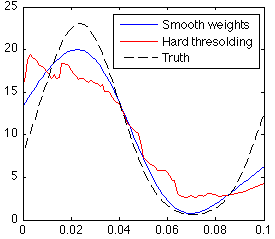}}
\caption{\label{asdf} {\it Left:} The simulated $\sigma(t) W(t)$ were $W$ is a mean zero Gaussian random field with Mat\'ern autocovariance function with parameters $(\sigma^2,\nu,\rho)=(1,0.8,0.2)$ and $\sigma(t)=2\sin(t/0.015)+2.8$. {\it Middle:} The increments of $\sigma(t)W(t)$ on the $200$ evenly spaced observation locations in the interval $[0,0.1]$.
{\it Right:} Plot of the true local parameter function $\sigma^2(t)$ (dashed) along with two estimates of $\sigma^2(t)$, one using smooth weights (blue) and the other using hard thresholding weights (i.e. na\"ive local likelihood estimation).}
\end{figure}

Figure \ref{asdf} shows an example of this local likelihood estimate  (blue line in the right-hand diagram) as compared with na\"ive local likelihood estimation (red line in the right-hand diagram) where the stationary model is presumed to hold through the local neighborhood.
 The left most diagram in Figure \ref{asdf} shows the simulated field $\sigma(t) W(t)$ at $200$ evenly spaced observation locations in $[0,1]$, where $W$ is a mean zero Gaussian random field with known Mat\'ern autocovariance function with parameters $(\sigma^2,\nu,\rho)=(1,0.8,0.2)$  (using the parameterization given on page 50 of \cite{stein:book}) and unknown $\sigma(t)=2\sin(t/0.015)+2.8$. Notice that even though the local variance changes significantly throughout the observation region, it is difficult to see since the range parameter $\rho=0.2$ is relatively large compared with the observation region.  However, if one looks at increments of $\sigma(t)W(t)$, shown in the middle diagram, the change  in local variance becomes clearly visible. The true  local parameter function $\sigma^2(t)$ is shown in the right hand plot of Figure \ref{asdf} (dashed).
The smooth weights take the form $w_k=K((t-t_k)/\lambda)$ where $K$ is defined as $K(t)\triangleq e^{-t^2/2}(15-10 t^2+ t^4)/8$ which is kernel $K_6$ presented in \cite{wand:90} (and also used in Section \ref{Bay} below).  
The bandwidth used for both estimators is obtained by minimizing the Kullback-Leibler divergence of the estimated global model $G_{\hat\sigma^2_\lambda(\cdot)}$ to the truth at the observation locations.
 Notice that using smooth weights, the local likelihood estimate can not only reduce estimation error but can also yield smooth estimates of the local parameter function $\sigma^2$.

 {\it Remark:}
Although the computation of each  $\Sigma_k^{-1}$ in  (\ref{softEst}) can be very time consuming, there is a downdating algorithm that makes all the inverses obtainable in nearly the same computational time as $\Sigma_n^{-1}$ (where $n$ is the number of observations).
 Once the inverse $\Sigma_{k}^{-1}$ is computed, the downdate $\Sigma_{k-1}^{-1}$ is easily computable by the following formula:
 \[ \Sigma_{k}^{-1}=\begin{pmatrix} \Sigma_{k-1}& b\\ b^T & c \end{pmatrix}^{-1} = \begin{pmatrix}
\Sigma_{k-1}^{-1} + \frac{\Sigma_{k-1}^{-1} bb^T\Sigma_{k-1}^{-1} }{k}& -\frac{\Sigma_{k-1}^{-1}b}{ k}\\ 
 -\frac{b^T\Sigma_{k-1}^{-1}}{ k} &\frac{1}{k}
\end{pmatrix}
 \]
 where $k=c-b^T\Sigma_{k-1}^{-1} b$. Therefore
 \[ \Sigma_{k-1}^{-1}= \Sigma_{k}^{-1}{(1\colon k-1,1\colon k-1)}-\frac{\Sigma_{k}^{-1}(1:k-1,k)\Sigma_{k}^{-1}(k,1:k-1)}{\Sigma_{k}^{-1}(k,k)}. \]
where $\Sigma_{k}^{-1}{(a\colon\! b,c\colon\! d)}$  denotes the sub-matrix of $\Sigma_{k}^{-1}$ of rows $a$ through $b$ and columns $c$ through $d$. Note that the estimates $\hat\sigma^2 (t)$ and  $\hat\sigma^2(s)$ at two different locations $t\neq s$ are very similar: the only difference is the order one downdates the  row and column of $\Sigma_n^{-1}$. This means one can construct $\hat \sigma^2$ at may different spatial locations by computing $\Sigma_n^{-1}$ once and for all, then downdate the inverse in different row and column orders.

 \subsection{Bayes risk under polynomial priors for $\sigma(\cdot)$}
 \label{Bay}
 
Now we study the behavior of $\hat \sigma_\lambda^2(t_0)$ at some fixed spatial location $t_0$. We think of the parameters that govern the deviation of $\sigma(t)$ from $\sigma(t_0)$ as nuisance parameters and model this situation by supposing $\sigma$ has the following finite taylor expansion
 \begin{equation}
 \label{taylor}
  \sigma(t)= c_0+\sum_{p=1}^N c_p (t-t_0)^k.
  \end{equation}
 The zero order coefficient, $c_0$, is the parameter of interest and the higher order coefficients, $c_p$ for $p>0$, are nuisance parameters. Now 
using a prior $\pi(\bs c)$ for the nuisance coefficients  $\bs c\triangleq (c_1,\ldots,c_N)$ we derive the Bayes risk  of $\hat\sigma_\lambda^2(t_0)$ under $L_2$ loss. For simplicity we suppose  $\pi(\bs c) =\pi_1(c_1)\cdots \pi_N(c_N)$, $E_{\pi_j} c_j=0$ and $E_{\pi_j} c^3_j=0$ for all $j\geq 1$. 
One of the goals of this section to attempt to understand the relationship between these nuisance parameters and using higher order kernels that are orthogonal to polynomials of order at most $N$. The other goal is to study the relationship between bias, risk, bandwidth and parameters of the stationary random field.

 Let $(t_1,\ldots,t_k)$ be the observation locations for $\bs z_k= (z_1,\ldots,z_n)^T$ so that $z_j=\sigma(t_j) W(t_j)$. To make the notation clearer we suppose the observations are ordered by their distance to $t_0$, so that $\mathcal N_{t_0,3}=\{z_1,z_2,z_3 \}$ for example.  Now $\bs z_k$ has covariance function $\Delta_\sigma \Sigma_{k} \Delta_{\sigma}$ where
 \begin{align}
  \Delta_\sigma&\triangleq\text{diag}(\sigma(t_1),\ldots,\sigma(t_n))
  =\sum_{p=0}^Nc_k \Delta_{n}^p
  \end{align}
   and $\Delta_{n}^p\triangleq\text{diag}[(t_1-t_0)^p,\ldots,(t_n-t_0)^p]$.     
 To derive expressions for the bias and mean square error for $\hat\sigma^2_\lambda(t_0)$, notice that
 \begin{align}
 E\bigl[\bs z_k^T \Sigma_k^{-1} \bs z_k\bigr|\bs c\bigr] &=  E\text{tr}\bigl[\bs z_k \bs z_k^T \Sigma_k^{-1}\bigr|\bs c\bigr] =\text{tr}[\Delta_{\sigma} \Sigma_k \Delta_\sigma \Sigma_k^{-1}] 
 = \sum_{p_1,p_2=0}^Nc_{p_1}c_{p_2} \text{tr} \left[ \Sigma_k^{-1} \Delta^{p_1}_k \Sigma_k \Delta_k^{p_2}   \right]
 \end{align}
 and 
 \begin{align}
   \text{cov}\bigl[\bs z_k^T \Sigma_k^{-1} \bs z_k,\bs z_j^T \Sigma_{j}^{-1}\bs z_j\bigr|\bs c\bigr]&=  \text{cov}\bigl[\bs z_k^T \Sigma_k^{-1} \bs z_k,\bs z_k^T \Sigma_{j\rightarrow k}^{-1}\bs z_k\bigr|\bs c\bigr] 
   =\text{tr}[\Sigma_k^{-1} \Delta_\sigma \Sigma_k \Delta_\sigma \Sigma_{j\rightarrow k}^{-1} \Delta_\sigma \Sigma_k \Delta_\sigma  ] \\
   &= 2 \sum_{p_1,\ldots,p_4=0}^N c_{p_1}\cdots c_{p_4} \text{tr}[\Sigma_k^{-1} \Delta_k^{p_1} \Sigma_k \Delta_k^{p_2} \Sigma_{j\rightarrow k}^{-1} \Delta_k^{p_3} \Sigma_k \Delta_k^{p_4}  ]
    \end{align}
where $k\geq j$ and $\tilde\Sigma^{-1}_{j\rightarrow k}\triangleq\left( {\Sigma^{-1}_{j} \atop 0} { 0 \atop 0}\right)$ is the matrix $\Sigma_j^{-1}$ padded with zeros so that it has the same size as $\Sigma_k^{-1}$. 
At this point it becomes convenient to re-write $\hat\sigma_\lambda^2(t_0) $ as $\sum_{k=1}^n \tilde w_k \; \bs z_k^T \Sigma_k^{-1} \bs z_k$ where 
\[\tilde w_k\triangleq \begin{cases} (w_k-w_{k+1})/\sum_{j=1}^n w_j, & \text{if $k<n$;} \\  w_n/\sum_{j=1}^n w_j, & \text{if $k=n$.} \end{cases}\] 
 Now the expected value of $\hat \sigma_\lambda^2(t_0)$, conditional on $\bs c= (c_1,\ldots,c_N)$, is 
 \begin{align}
 E\bigl[\hat \sigma_\lambda^2(t_0)\bigr|\bs c\bigr]&=\sum_{k=1}^n \tilde w_k E\bigl[\bs z_k^T \Sigma_k^{-1}\bs z_k\bigr|\bs c\bigr]  =\sum_{k=1}^n\sum_{p_1,p_2=0}^N \tilde w_k c_{p_1}c_{p_2}  \text{tr}\bigl[\Sigma_k^{-1} \Delta_{k}^{p_1} \Sigma_k \Delta_{k}^{p_2} \bigr] =\sum_{p_1,p_2=0}^N c_{p_1}c_{p_2} B^{p_1,p_2}
  \end{align}
  where $ B^{p_1,p_2}\triangleq \sum_{k=1}^n\tilde w_k  \text{tr}\Bigl[ \Sigma_k^{-1} \Delta_{k}^{p_1} \Sigma_k \Delta_{k}^{p_2}  \Bigr]$.
The variance is
 \begin{align}
 \text{var}\bigl[\hat \sigma_\lambda^2(t_0)|\bs c\bigr] &=  \text{var}\left[\sum_{k=1}^n \tilde w_k\,  \bs y^T_k \Sigma_k^{-1} \bs z_k\Bigl|\bs c\right] =  \sum_{j,k= 1}^n \tilde w_k \tilde w_j \,\text{cov}[\bs y^T_k \Sigma_k^{-1} \bs z_k, \bs y^T_j \Sigma_j^{-1} \bs z_j] \\
  &= 2 \sum_{j,k= 1}^n  \sum_{p_1,\ldots,p_4=0}^{N} \tilde w_k \tilde w_j c_{p_1}\cdots c_{p_4} \text{tr}[\Sigma_k^{-1} \Delta_k^{p_1} \Sigma_k \Delta_k^{p_2} \Sigma_{j\wedge k\rightarrow j\vee k}^{-1} \Delta_k^{p_3} \Sigma_k \Delta_k^{p_4}  ] \\
  &=  \sum_{p_1,\ldots,p_4=0}^{N} c_{p_1}\cdots c_{p_4} B^{p_1,\ldots,p_4}  \end{align}
  where
$B^{p_1,\ldots,p_4} \triangleq   2 \sum_{j,k= 1}^n \tilde w_k \tilde w_j \text{tr}\Bigl[\Sigma_k^{-1} \Delta_k^{p_1} \Sigma_k \Delta_k^{p_2} \Sigma_{j\wedge k\rightarrow j\vee k}^{-1} \Delta_k^{p_3} \Sigma_k \Delta_k^{p_4} \Bigr ] $.
 Now the expected squared bias under the prior distribution for the nuisance parameters can be decomposed into two quadratic terms
 \begin{align}
 E_{\pi}[ \text{bias}^2] &= E_{\pi}[  E\bigl[\hat \sigma^2(t_0)|\bs c\bigr] -\sigma^2(t_0)]^2 =  E_{\pi}\left[ \sum_{p_1,p_2=0}^N c_{p_1}c_{p_2} B^{p_1,p_2} -c_0^2\right]^2 \\
&=  E_{\pi}\left[ \sum_{p_1,p_2=1}^N c_{p_1}c_{p_2} B^{p_1,p_2} + 2 c_0\sum_{p=1}^N c_{p} B^{0,p} \right]^2 \\
&=E_{\pi}\left[ \sum_{p_1,p_2=1}^N c_{p_1}c_{p_2} B^{p_1,p_2}\right]^2 + 4\, c_0^2\, E_{\pi}\left[\sum_{p=1}^N c_{p} B^{0,p}\right]^2 \label{lasteeq}
 \end{align}
since $B^{0,0}=\sum_{k=1}^n k \tilde w_k=1$ and $B^{0,p}=B^{p,0}$. Note that the cross term in (\ref{lasteeq}) is zero by the assumption $E_{\pi_j} c^3_j=0$ for all $j\geq 1$.
 Therefore the Bayes risk is
 \begin{align} 
 \label{BayesR}
 E_{\pi} E \Bigl[\hat \sigma^2(t_0)- {\sigma^2(t_0)}  \bigr]^2\Bigl|\bs c\Bigr] &=   \sum_{p_1,p_2,p_3,p_4=0}^N E_{\pi}\bigl[ c_{p_1}c_{p_2}c_{p_3} c_{p_4}\bigr] B^{p_1,p_2,p_3,p_4} \\
 &\qquad+  \sum_{p_1,p_2,p_3,p_4=1}^N E_{\pi}\bigl[ c_{p_1}c_{p_2}c_{p_3} c_{p_4}\bigr] B^{p_1,p_2}B^{p_3,p_4} \\ \label{lastTerm}
 &\qquad\qquad+ 4\, c_0^2\, E_{\pi}\left[\sum_{p=1}^N c_{p} B^{0,p}\right]^2.
  \end{align}
    Notice that the above expression decomposes the Bayes risk into three quadratic terms.
The last term (\ref{lastTerm}), which comes from the bias, can be rewritten as
   \begin{align}
  B^{0,p}&= \sum_{k=1}^n\tilde w_k  \text{tr}\bigl[ \Delta_{k}^{p}  \bigr] =\sum_{k=1}^n\tilde w_k  \text{tr}\bigl[ \text{diag}((t_1-t_0)^p,\ldots,(t_k-t_0)^p) \bigr] \\
  &=\sum_{k=1}^n\tilde w_k [(t_1-t_0)^p+\ldots+(t_k-t_0)^p]\\
  &=\frac{1}{\sum_{j=1}^n w_j}\sum_{k=1}^n w_k  (t_k-t_0)^p.
  \end{align}
Therefore, if  $w_k$ is constructed by $w_k\triangleq  K((t_0-t_k)/\lambda)$, where $\int_{\Bbb R} K(t) dt= 1$ and $\int_{\Bbb R} K(t) t^p dt= 0$ for some $p>1$, then under mild conditions on the sampling locations $t_k$ and $K$
\begin{align*}
 B^{0,p}&=\frac{\sum_{k=1}^n K\left[\frac{t_0-t_k}{\lambda}\right]  (t_k-t_0)^p}{ \sum_{j=1}^n K\left[\frac{t_0-t_j}{\lambda}\right] }\longrightarrow 0
 \end{align*}
as $n\rightarrow \infty$, $\lambda \rightarrow 0$ and the $t_k$'s get more dense in a bounded region near $t_0$.   In particular, if one uses higher order kernels, the last term (\ref{lastTerm}) can be made arbitrarily small under infill asymptotics. Note: it is only necessary to consider $\lambda\rightarrow 0$ when our observation locations $t_k$ stay bounded in a compact domain and $K$ has infinite support (which is the case we will consider in this section).

\begin{figure}[h]
\centering
{\includegraphics[height=2.4in]{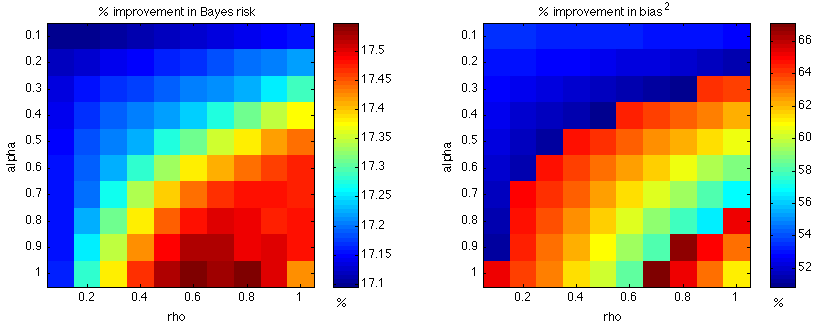}}
\caption{ \label{heat} The percent improvement in Bayes risk (left) and squared bias (right) over hard threshold when using an oracle bandwidth selector and weights generated from kernel $K_6$. The rows of the heat maps correspond to different values of $\nu$ and the columns to $\rho$ for the Mat\'ern autocovariance function. }
\end{figure}

\begin{figure}[h]
\centering
{\includegraphics[height=2.4in]{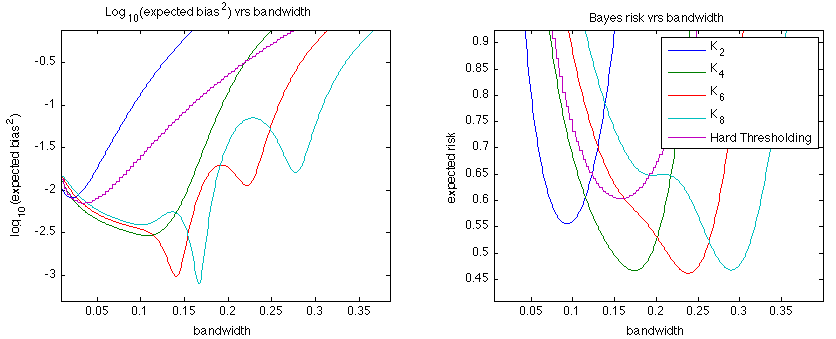}}
\caption{\label{Rvb}
Bayes risk (right) and expected bias squared (left) of $\hat\sigma^2_\lambda(1/2)$ plotted against bandwidth for $5$ different kernels: $K_{2}$, $K_4$, $K_6$, $K_8$ and hard thresholding. 
The estimate $\hat\sigma_\lambda^2(1/2)$  is based on 150 observation locations in $[0,1]$ and the Mat\'ern autocovariance function with $\nu=\rho=0.8$.}
\end{figure}

\paragraph{Numerical Results} 
 For exposition we consider a class of Gaussian-based higher order kernels $K_{2r}$ defined in \cite{wand:90} as
 \begin{equation}\label{highK} K_{2r}(t)\triangleq Q_{2r-2}(t) \phi(t) \end{equation}
where $\phi(t)=e^{-t^2/2}/\sqrt{2\pi}$ is the Gaussian kernel, $Q_{2r-2}=\{2^{r-1}(r-1)! \}^{-1} H_{2r-1}(t)/t$, and $H_j$ denotes the $j^\text{th}$ normalized Hermite polynomial defined by $H_j(t)=(-1)^j \phi^{(j)}(t)/\phi(t)$. These kernels have the required property that $\int_{\Bbb R} K(t) t^p dt= 0$ for all $0<p< 2r$. Of course these are not the only such kernels but we use them since the resulting estimates $\hat\sigma_\lambda^2(t)$ will be very smooth in $t$. 
 Finally we compare all our results to the the hard thresholding weights defined as
 $ w_k \triangleq  \bs 1_{B_\lambda(t)}(t_k)$,
 where $\bs 1_A$ is the indicator of the set $A\in \Bbb R$ and $B_\lambda(t)$ is the ball of radius $\lambda$ centered at $t$. The hard thresholding weights yield the na\"ive local likelihood estimate where the local stationary model is presumed to hold everywhere in the local neighborhood.

Now we use the above derivation of Bayes risk to numerically investigate the improvement of local likelihood estimation over hard thresholding. We computed the Bayes risk and bias  terms in equation (\ref{BayesR}) under the assumption that $t_0=1/2$, $\sigma(t_0)=2$, and $c_j\overset{iid}\sim \mathcal N(0,4)$ for $j=1,\ldots,4$ (so that $N=4$ in (\ref{taylor})).
The heat maps in Figure \ref{heat} show the percent improvement over hard threshold when using  kernel $K_6$ defined in (\ref{highK}) to generate the weights $w_k$.  The bandwidth for both $K_6$  and hard thresholding  were chosen using the Kullback-Leibler oracle criterion. The rows of the heat maps correspond to different values of $\nu$ and the columns to $\rho$ for the Mat\'ern autocovariance function. The observation locations for which the local likelihood estimate $\hat\sigma_\lambda$ is based on  are $100$ evenly spaced points in $[0,1)$. First notice that there is a 17\% improvement over hard thresholding uniformly over the possible values of $\nu$ and $\rho$. An interesting feature of the leftmost diagram is that local likelihood estimates do better as the random field gets smoother. A possible explanation is that the smoother the random field the larger the neighborhood required to attain a sufficient variance reduction. For these large neighborhoods the higher order kernels have significantly less bias. Finally, we mention that the right-hand diagram shows an improvement in expected bias squared that can reach almost 68\%.

 In Figure \ref{Rvb} we plot the Bayes risk and expected bias squared as a function of bandwidth for the estimate $\hat\sigma_\lambda^2(1/2)$ using the kernels   $K_{2}$, $K_4$, $K_6$, $K_8$  and hard thresholding. The estimate is based on the Mat\'ern random field with parameters $\nu=\rho=0.8$ with $150$ even observation locations in $[0,1]$. We use the same polynomial prior as in the last paragraph. Notice the significant improvement in risk and bias squared using kernels $K_4$, $K_6$ and $K_8$. Also notice that the improvement in risk and bias occurs at larger bandwidths for higher order kernels.

%

%

\section{Bandwidth selection}
\label{Band}
The local likelihood estimates $\hat\theta_\lambda$ can vary dramatically depending on the choice of smoothing parameter $\lambda$. Since the theoretical derivation of Bayes risk in anything but the simplest setting is extremely difficult it is necessary to develop estimates of $\lambda$ from data. Classical methods such as cross-validation may fail dramatically when the data is comprised of a single realization of a nonstationary random field. The data is very highly correlated so that a \lq\lq leave out" prediction is problematic since the data \lq\lq left out" is highly correlated with the data \lq\lq left in". In this section we give a heuristic for constructing a reasonable estimate $\hat \lambda$. We then present our interpretation of this heuristic in terms of two different estimates of $\lambda$. At the end of this section we present some numerical simulations to illustrate their behavior. We make no claim of optimality or any theoretical justification other than heuristics. 

Our heuristic for the bandwidth  estimator says that one should choose the bandwidth to maximize the spatial variability in $\hat\theta_\lambda$ beyond what is expected from the realization of the random field itself. The idea is that  there are two sources of spatial variability in $\hat \theta_\lambda$. The first is the spatial variability from the true local parameter function $\theta$ which, of course, $\hat\theta_\lambda$ is trying to estimate. The second is the spatial variation coming from the particular realization of the random field itself. Notice that the smaller the bandwidth $\lambda$ the more the spatial variation in $\hat\theta_\lambda$ is due to the random field realization, and less to the variability of  $\theta$. 
To describe our mathematical interpretation of this heuristic we need some notation.  Let  $\mathcal P(\hat\theta_\lambda)$ be a measure of spatial variation in $\hat\theta_\lambda$. For example, when the local parameter $\theta$ is univariate so the  $\theta(t)$ maps $\Bbb R^d$ to $\Bbb R$ a natural choice might be $\mathcal P( \hat\theta_\lambda )\triangleq \int_{\Bbb R^d}|\nabla \hat\theta_\lambda |^2 $. 
Let $\overline \theta$ be the MLE of the local parameter function obtained by assuming the global model  is stationary (i.e.\! by assuming  $\theta(t)$ is a constant function of $t$). Now consider simulating an independent realization of the data under the stationary random field model using the estimated parameter $\overline\theta$. Let $E_{{\overline \theta}} \mathcal P( \hat\theta_\lambda )$ denote the expected value of $\mathcal P(\hat\theta_\lambda)$ where $\hat\theta_\lambda$ is applied to the new realization of the stationary random field. The idea is that $E_{{\overline \theta}} \mathcal P( \hat\theta_\lambda )$ quantifies  the spatial variation in $\hat \theta_{\lambda}$ that is exclusively due  to the random field itself (since the true local parameter function is constant). Now we estimate $\lambda$ as follows:
\begin{equation}
\label{H1}
\hat\lambda_\text{1} \triangleq  \text{arg}\max_\lambda \frac{\mathcal P( \hat\theta_\lambda )- E_{{\overline \theta}} \mathcal P( \hat\theta_\lambda )}{{\text{sd}}_{{\overline \theta}} \mathcal P(\hat\theta_\lambda )} .
\end{equation}
In equation (\ref{H1}) we divided by ${\text{sd}}_{{\overline \theta}} \mathcal P(\hat\theta_\lambda )$, which denotes the standard deviation of $\mathcal P(\hat\theta_\lambda )$ under the stationary model given by $\overline \theta$, to improve comparison across $\lambda$.

Our second bandwidth selector is similar to the first, the main difference being that the measure of spatial variation $\mathcal P(\hat\theta_\lambda)$ is replaced with $\sum_t ( \mathcal W_\lambda (\hat\theta_\lambda(t),t) - \mathcal W_\lambda (\overline\theta,t))$.  This is recognized as a type of local likelihood ratio test statistic (summed over the spatial variable $t$) that compares the estimate $\hat\theta_\lambda$ to the stationary fit $\overline \theta$. Then after adjusting by the expected behavior of this quantity under the stationary fit $\overline \theta$ we get the following estimate of bandwidth:
\begin{equation}
\label{secondBand}
 \hat\lambda_2=\text{arg}\max_{\lambda}\frac{\sum_t ( \mathcal W_\lambda (\hat\theta_\lambda(t),t) - \mathcal W_\lambda (\overline\theta,t))-E_{\overline \theta}\left[\sum_t ( \mathcal W_\lambda (\hat\theta_\lambda(t),t) - \mathcal W_\lambda (\overline\theta,t)) \right]}{\text{sd}_{\overline \theta}\left[\sum_t ( \mathcal W_\lambda (\hat\theta_\lambda(t),t) - \mathcal W_\lambda (\overline\theta,t)) \right]}
 \end{equation}
 In our implementation of both (\ref{H1}) and (\ref{secondBand}) we use simulations to estimate the expected value and standard deviation under the stationary fit $\overline \theta$.
 
\begin{figure}[h]
\centering
{\includegraphics[height=2.4in]{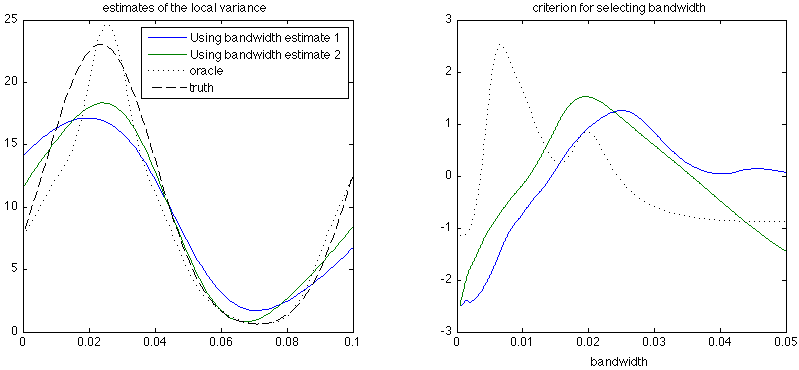}}
\caption{\label{SNRGood2}
{\it Left:} $\sigma^2$ (dashed), $\hat\sigma^2_{\hat\lambda_1}$ (blue), $\hat\sigma^2_{\hat\lambda_2}$ (green) and  $\hat\sigma^2_{\lambda_\text{orc}}$ (dotted) when observing a single realization of $\sigma(t) W(t)$, where $\sigma(t)=2\sin(t/0.015)+2.8$ is unknown and $W$ is a stationary Gaussian random field with known Mat\'ern parameters $(\sigma, \nu,\rho)=(1,0.5,0.5)$, at 1000 even sampling locations in $[0,0.1]$. {\it Right:}  Plots of the standardized criterion profiles which, when maximized, give $\lambda_\text{orc}$ (dotted), $\hat\lambda_{1}$ (blue) and $\hat\lambda_{2}$ (green). }
\end{figure}

\begin{figure}[h]
\centering
{\includegraphics[height=2.4in]{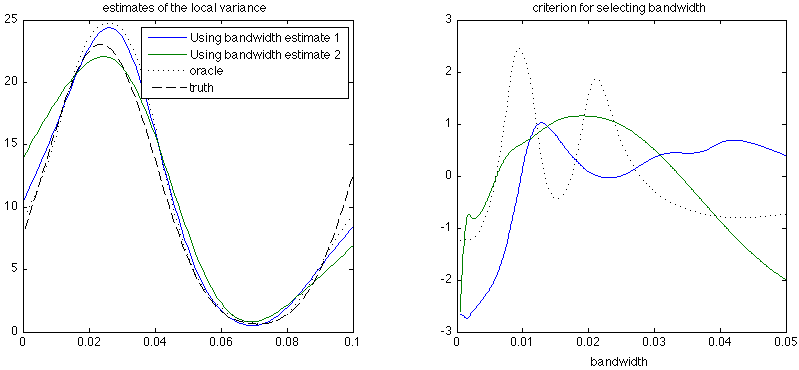}}
\caption{\label{SNRGood22}
Estimates $\hat\sigma^2_{\hat\lambda_1}$, $\hat\sigma^2_{\hat\lambda_2}$ and  $\hat\sigma^2_{\lambda_\text{orc}}$  of $\sigma^2$ (left) along with the criterion profiles for estimating $\lambda$ (right). The simulation parameters are the same as in Figure \ref{SNRGood2} with the exception that $\nu=1$ instead of $\nu=0.5$.
}
\end{figure}

To illustrate $\hat\lambda_\text{1}$ and $\hat\lambda_2$ we provide some simulations in the case of variance modulation (see Section \ref{VM}). Our first simulation is a single realization of $\sigma(t)W(t)$ at $1000$ evenly spaced sampling locations in the interval $[0,0.1]$  where $\sigma(t)=2\sin(t/0.015)+2.8$ is unknown and $W$ is a mean zero stationary Gaussian random field with known Mat\'ern autocovariance parameters $(\sigma,\nu,\rho)=(1,0.5,0.5)$. The left plot of Figure \ref{SNRGood2} shows the true local parameter function $\sigma^2(t)$ (dashed line) along with three different estimates $\hat\sigma^2_{\hat\lambda_1}$ (blue), $\hat\sigma^2_{\hat\lambda_2}$ (green) and  $\hat\sigma^2_{\lambda_\text{orc}}$ (dotted line). The estimate $\hat\sigma^2_{\lambda_\text{orc}}$ is the oracle estimate,   which estimates the smoothness parameter by 
\begin{equation}\lambda_\text{orc}\triangleq \text{arg}\max_{\lambda} \bigl\{D( G_{\hat\sigma_\lambda^2} \| G_{\sigma^2} )^{-1}\bigr\},\end{equation}
where $G_{\sigma^2}$ denotes the true global model for $\sigma(t)W(t)$  at the observation locations and $D(\cdot \| \cdot)$ is Kullback-Leibler divergence.
The right plot of Figure \ref{SNRGood2} shows the criterion profiles that are maximized for estimating the bandwidths $\hat \lambda_1$ (blue), $\hat\lambda_2$ (green) and $\lambda_\text{orc}$ (dotted). The profiles in this diagram  are standardized by their average and standard deviation so they can be compared on the same scale.
The parameters of our second simulation are exactly the same with the exception that now $\nu=1$ instead of $\nu=0.5$.  The results of this simulation are shown in Figure \ref{SNRGood22}, where again, the left plot shows $\sigma^2$ (dashed), $\hat\sigma^2_{\hat\lambda_1}$ (blue), $\hat\sigma^2_{\hat\lambda_2}$ (green) and  $\hat\sigma^2_{\lambda_\text{orc}}$ (dotted). The right plot, again shows the standardized profiles for estimating $\lambda$. 

For both simulations, $\hat\lambda_1$ and $\hat\lambda_2$ do a good job of estimating an appropriate bandwidth. The modes of the data-driven criterion profiles seem to come reasonably close to at least one of the modes in the oracle profiles. It is unclear to us, at this moment, why the oracle profile in both simulations have two modes. Regardless, we think it is a good sign that the data-driven profile modes are close to at least one oracle mode.  Moreover the resulting estimates $\hat\sigma^2_{\hat\lambda_1}$ and $\hat\sigma^2_{\hat\lambda_2}$ give good agreement with  the truth $\sigma^2$, although some under-fitting can be seen in Figure \ref{SNRGood2}.  Notice that  both of these simulations  use a large range parameter $\rho=0.5$ as compared to the size of the observation region $[0,0.1]$. This was intentional and designed to mimic the infill asymptotic regime of letting $\lambda\rightarrow 0$ as one gets denser observations.

\section{Estimating a local fractional index in 2 dimensions.}
\label{2dExample}
In this section we take full advantage of our local likelihood technique and estimate a spatially varying smoothness parameter when observing one realization of the field at random observation locations in $[0,1]^2$.  In our simulation we suppose that there is {\it a priori} information that the observed random field is locally isotropic with known local variance and range that do not vary spatially but with an unknown spatially varying smoothness parameter, denoted $\nu_t$ ($t\in\Bbb R^2$). There are many nonstationary random fields that have spatially varying smoothness parameter (see \cite{peltier:95}, \cite{stoev:06}, \cite{1259022} and \cite{sly:07} for example). For our simulation we use a nonstationary local Mat\'ern, presented below that  has a closed form covariance, can attain any degree of differentiability or H\"older smoothness and behaves locally like a stationary Mat\'ern. We include an appendix that proves the positive definiteness of our covariance function (which originally appeared in the technical report \cite{stein:nonstatTech}).

We start by deriving our nonstationary covariance structure with spatially varying local parameters.
Let  $\mathcal M_\nu(\cdot)\triangleq  |\cdot|^\nu \mathcal K_{\nu}( \cdot)$  where $\mathcal K_\nu$ is the modified Bessel function of the second kind and let $\nu_t$, $\sigma_t$ and $\alpha_t$ denote spatially varying Mat\'ern parameters. The local parameter function $\nu_t\colon \Bbb R^d\rightarrow \Bbb R^+$ determines the local smoothness, $\sigma^2_t\colon \Bbb R^d\rightarrow \Bbb R^+$ determines local variance and $\alpha_t\colon \Bbb R^d\rightarrow PD_d(\Bbb R)$ determines a local geometric anisotropy which maps $\Bbb R^d$ into the set of $d\times d$ positive definite matrices with real entries. Now define the following covariance function
\begin{equation}
\label{nonstatt}
 R(s,t)= \sigma_t\sigma_s \det\bigl({\alpha_{st}^{-1/2}}\bigr) \mathcal M_{\nu_{st}}\left(\bigl|\alpha_{st}^{-1/2} (t-s)\bigr|\right) 
 \end{equation}
where $\alpha_{st}\triangleq(\alpha_t+\alpha_s)/2$, $\nu_{st}\triangleq (\nu_t+\nu_s)/2$. 
In the appendix we show that the  nonstationary covariance function $R$ is positive definite.

One problematic feature about $R$ is that it is difficult to separate the interpretation of the local scale $ \det(\alpha_t^{-1/2})$ and the smoothness parameter $\nu_t$.  This is because both $\det(\alpha_t^{-1/2})$ and $\nu_t$ effect the local lag for which the correlation becomes close to zero.  Therefore it is desirable to re-parameterize (\ref{nonstatt}) to give distinct interpretations of the three local parameter functions: variance, range and smoothness.  To simplify the exposition we suppose there is no geometric anisotropy so that $\alpha_t$ maps $t\in \Bbb R^d$ into the set of positively scaled $d\times d$ identity matrices. In particular let $\rho_t\colon \Bbb R^d\rightarrow \Bbb R^+$ and define
\[ 
K(t,s)\triangleq
\sigma_s\sigma_t \left[\frac{(\rho_s^2/4\nu_s)^{d/2}} {\Gamma(\nu_s) 2^{\nu_s-1}} \right]^{1/2} \left[\frac{(\rho_t^2/4\nu_t)^{d/2}} {\Gamma(\nu_t) 2^{\nu_t-1}} \right]^{1/2} \left[ \frac{\rho_s^2}{8\nu_s} + \frac{\rho_t^2}{8\nu_t} \right]^{-d/2} \mathcal M_{\nu_{st}} 
\left[ \left(\frac{\rho_s^2}{8\nu_s} + \frac{\rho_t^2}{8\nu_t}\right)^{-1/2} \bigl|s-t\bigr| \right].
\]
The advantage of this covariance function (which is a re-parameterization and simplification of (\ref{nonstatt}))  is that when both $s,t$ are near some fixed $t_0$,  and the local parameter functions $\nu_t$, $\rho_t$ and $\sigma^2_t$ are sufficiently smooth, we have that
\[ K(s,t)\approx \frac{\sigma^2_{t_0}}{\Gamma(\nu_{t_0})2^{\nu_{t_0}-1}} \mathcal M_{\nu_{t_0}}(2\sqrt{\nu_{t_0}} |s-t|/\rho_{t_0} ).\]
This is recognized as an isotropic Mat\'ern autocovariance with parameterization found on page 50 of \cite{stein:book}.
Therefore $\sigma_t^2$ is the local variance, $\nu_t$ is the local smoothness, and $\rho_t$ has the interpretation of the local range. Now our {\it a priori} information on the observed random field  amounts to the supposition that $\rho_t\equiv \rho$ and $\sigma_t\equiv \sigma$ are constant functions of $t\in\Bbb R^d$ so that $K$ simplifies to 
\begin{equation}
\label{SSim}
K(s,t)=\sigma^2 \frac{\nu_s^{d/4} \nu_t^{d/4}}{\nu_{st}^{d/2}(\Gamma(\nu_s ) 2^{\nu_s-1} )^{1/2}(\Gamma(\nu_t ) 2^{\nu_t-1} )^{1/2}} \mathcal M_{\nu_{st}}\left[ \frac{2\sqrt{\nu_s\nu_t}}{\rho\sqrt{\nu_{st}}}|s-t| \right]. 
\end{equation}

When estimating a local parameter function in two dimensions,  bias becomes prominent near the boundary of the observation region. In our simulation we attempted to design the local likelihood weights to automatically mitigate this bias near the boundary.  We can loosely motivate our weights by the theory of estimating equations. Notice that if the local parameter function $\theta$ is univariate (so that it maps $\Bbb R^d$ into $\Bbb R$)  then $\frac{d}{d\theta_0}\mathcal W_\lambda (\theta_0,t|\text{data})=\sum_{k=1}^n w_k S^\prime_k (\theta_0)$ where $S_k(\theta_0)=\log f_{\theta_0}(z_k|\mathcal N_{k-1,t})$ is the conditional log likelihood of the $k^\text{th}$ nearest observation to $t$ conditional on  $k-1$ nearer observations. Therefore the local likelihood estimate $\hat\theta(t)$ is defined as solving 
\[ \sum_{k=1}^n w_k S_k^\prime (\hat\theta(t))=0 \]
where both $S_k$ and $w_k$ depend on $t$.
This suggests to design the weights $w_k$ to minimize  $\text{var}_\theta\bigl[\sum_{k=1}^n w_k S_k^\prime (\theta_0) ]$ subject to the unbiasedness constraint  $\sum_{k=1}^n w_k E_\theta(S_k^\prime (\theta_0))=0$ where $\theta(t)=\theta_0$. Notice that $E_\theta$ and $\text{var}_\theta$ denote expected value and variance with respect to the true nonstationary model $G_{\theta(\cdot)}$. To exclude the solution $w_1=\cdots =w_n=0$ we fix the scale $\sum_{k=1}^n w_k=1$. 
Now if we ignore the covariance of the cross terms $\text{cov}(w_k S_k^\prime (\theta_0),w_j S_j^\prime( \theta_0))$ (which are zero under the stationary model since $S_1^\prime (\theta_0), \ldots, S_n^\prime (\theta_0)$ then forms a martingale difference sequence) one wishes to design the weights $w_k$  to minimize $ \sum  w^2_k \, \text{var}_\theta [S^\prime_k (\theta_0)] $ subject to $\sum_{k=1}^n w_k E_\theta(S_k^\prime (\theta_0))=0$. 

First notice that the quantity $1/\text{var}_\theta [S^\prime_k (\theta_0)] $ can be interpreted as the  information  that the data $z_k$ provides for $\theta_0$, conditional on the $k-1$ nearer observations to $t$. In the fully nonstationary case we stipulate that $1/\text{var}_\theta [S^\prime_k (\theta_0)] $ will decay to zero as $t_k\rightarrow \infty$  and depend predominantly on the nuisance parameters that govern the deviation of $\theta(t_k)$ from $\theta(t)$ (e.g. on the coefficients of order $\geq 1$ of the Taylor expansion of $\theta(t_k)$ at $t$).
For a particular estimation problem it may be possible  to use prior knowledge on the scale of spatial variability of the local parameter function $\theta$ to attempt to get some understanding of how $1/\text{var}_\theta [S^\prime_k (\theta_0)] $ behaves as a function of $t_k$. In our example we model $1/\text{var}_\theta [S^\prime_k (\theta_0)] $ as $\exp [-|t-t_k|^2/2\lambda^2  ]$ for some unknown univariate  nuisance parameter $\lambda$, which is essentially plays the roll of a bandwidth parameter.
To approximate the unbiasedness requirement $\sum_{k=1}^n w_k E_\theta(S_k^\prime (\theta_0))=0$ we notice that  $E_\theta(S_k^\prime (\theta_0))$ is a function of $t,t_1,\ldots,t_k$ and converges to zero as $t_k\rightarrow t$. Moreover we believe its behavior will depend mostly on $t_k$ and $t$ so that the first order Taylor expansion in $t_k$ at $t$ is $E_\theta(S_k^\prime (\theta_0))\approx c_{t,\theta} (t-t_k)$ for an unknown constant $c_{t,\theta}$. 
This is our motivation for the following variational characterization of the local likelihood weights:
\begin{equation}
\label{WWeights} 
\text{minimize} \sum_{k=1}^n  w^2_k \exp [|t-t_k|^2/2\lambda^2  ] \text{  subject to  } \begin{cases}
\sum_{k=1}^n w_k=1, \\
\sum_{k=1}^n w_k (t-t_k)=(0,0)^T.
\end{cases}
\end{equation}
Notice that this particular variational characterization is essentially the same as in the nonparameteric regression setting (see \cite{muller:1987}). 

{\it Remark:} It is interesting to note that the solution of (\ref{WWeights}) is in the form
\[ w_k =[a + b\cdot (t-t_k)] \exp[-|t-t_k|/2\lambda^2]  \]
where $a$, $b$ are essentially the lagrange multipliers of the variational problem. Therefore the solution of $w_k$ in (\ref{WWeights}) can be computed quickly by inversion of a $2\times 2$ matrix for each $\lambda$.

\begin{figure}[h]
\centering
{\includegraphics[height=2in]{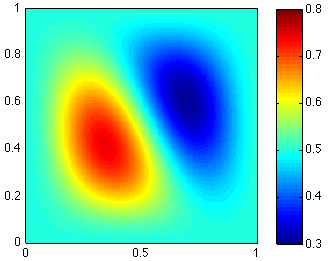}} 
{\includegraphics[height=2in]{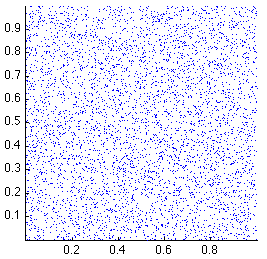}} \\
{\includegraphics[height=2in]{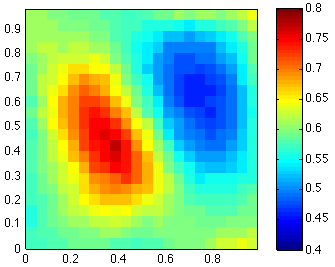}} 
{\includegraphics[height=2in]{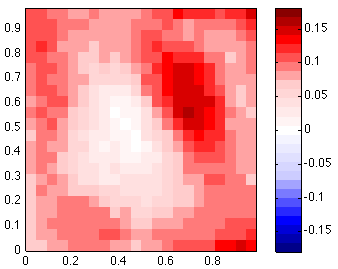}} 
\caption{\label{samplFig1}{\it Top left:}  The true local parameter function $\nu_t$ used in the simulation. {\it Top right:} The 5000 observation locations. {\it Bottom left:} The estimate of the local parameter function $\nu_t$ on the grid of $500$ estimation locations.  {\it Bottom right:} The error $\hat\nu_t - \nu_t$.}
\end{figure}

To test our estimate of $\nu_t$ we simulated one realization of a mean zero Gaussian random field  at 5000 random observation locations in $[0,1]^2$ using autocovariance function (\ref{SSim}) where $\rho=0.5$ and $\sigma=1$ are assumed to be known. We used the algorithm for Cholesky updating found in \cite{cholaug:02} to compute the log-likelihoods $\ell(\mathcal N_{t,k}|\nu)$ for $k=1,\ldots,500$ ($500$ was chosen so the computations could be done in reasonable time). A plot of the local parameter function $\nu_t$ and the sampling locations are shown on the top row of diagrams in Figure \ref{samplFig1}. 
To adjust for the fact that there is $\sim$75\% less data at the corners, and $\sim$50\% near the edges, we scaled the bandwidth by $2$ near the corner and by $\sqrt{2}$ near the edges  (by linear interpolation). The bottom left plot of Figure \ref{samplFig1} shows the estimate $\hat\nu_t$ on a square grid of $500$ estimation locations in $[0,1]^2$ using the bandwidth estimator $\hat\lambda_2$. The right plot shows the bias: $\hat\nu_t-\nu_t$. One can see that the estimate $\hat\nu_t$ does a good job of recovering the main features of the true local parameter function $\nu_t$. However, there is a positive bias all throughout the observation region with more bias in the region with  small $\nu_t$.

\appendix
\section*{Appendix}
\label{aapend1}

In this appendix we show that the nonstationary covariance function (\ref{nonstatt}) is positive definite. Originally presented in \cite{stein:nonstatTech}, the proof combines the results found in \cite{paciorek:04b}, \cite{paciorek:04a} for generating a spatially varying geometric anisotropy with those found in \cite{pintore:04} for generating a local isotropic Mat\'ern with spatially varying smoothness parameter.
 We start by stating the following lemma which  is proved using a convolution argument found in Paciorek's thesis \cite{paciorek:thesis} on, p. 27 .
\begin{lemma} Let $\alpha_t$ map $t\in\Bbb R^d$ to the set of real $d\times d$ positive definite matrices (denoted $PD_d(\Bbb R)$) and let $\phi(u)\triangleq \exp(-|u|^2/2)$ Then
\begin{equation}
\label{sshow}
\det(\alpha_{st}^{-1/2}) \phi\bigl[\alpha_{st}^{-1/2}(s-t)/\sigma\bigr] = c_t c_s \int_{\Bbb R^d}  \phi \bigl[\alpha_s^{-1/2}(u-s)/\sigma\bigr] \phi\bigl[\alpha_t^{-1/2}(u-t)/\sigma\bigr] du 
\end{equation}
where $\alpha_{st}\triangleq (\alpha_s + \alpha_t)/2$ and $c_s\triangleq (2\pi)^{-d/4}\sigma^{-d/2} \det(\alpha_t^{-1/2})$.
\end{lemma}
\noindent
The importance of this lemma is that  the right hand side of (\ref{sshow}) is positive definite with locally varying geometric anisotropy $\alpha_t^{-1/2}$. 
This follows since the right hand side of equation (\ref{sshow}) equals $\text{cov} (Z_\sigma(s), Z_\sigma(t))$ where $Z_\sigma(t)=c_t\int_{\Bbb R^d}   \phi\bigl[\alpha_t^{-1/2}(u-t)/\sigma\bigr] dW(u)$ and $dW$ is Gaussian white noise. 
Now let $K_\sigma(s,t)\triangleq\det(\alpha_{st}^{-1/2}) \phi\bigl[\alpha_{st}^{-1/2}(s-t)/\sigma\bigr]$ and 
notice that for any function $g_\sigma(\cdot)$ the function $g_\sigma(s)g_\sigma(t) K_\sigma(s,t)$ is positive definite. Since convex combinations and limits of positive definite functions are again positive definite we have that
\begin{equation}
\label{conveq}
 \int_0^\infty g_\sigma(s)g_\sigma(t) K_\sigma(s,t) d\mu(\sigma) 
 \end{equation}
is positive definite for any positive finite measure $\mu$ such that $\int_{\Bbb R^d }g^2_\sigma(t) K_\sigma(t,t)d\mu(\sigma)<\infty $ for all $t$.

\begin{claim}  Let $\sigma_t \colon \Bbb R^d\rightarrow \Bbb R^+$, $\nu_t  \colon \Bbb R^d\rightarrow \Bbb R^+$ and $\alpha_t \colon \Bbb R^d\rightarrow PD_d(\Bbb R)$ and define   $\alpha_{st}\triangleq (\alpha_t+\alpha_s)/2$, $\nu_{st}\triangleq (\nu_t+\nu_s)/2$. Then 
\[ R(s,t)= \sigma_t\sigma_s \det\bigl({\alpha_{st}^{-1/2}}\bigr) \mathcal M_{\nu_{st}}\left(\bigl|\alpha_{st}^{-1/2} (s-t)\bigr|\right) \]
is positive definite on $\Bbb R^d$ where $\mathcal M_\nu(x)=x^\nu \mathcal K_\nu(x)$ and $\mathcal K_\nu$ is the modified Bessel function of the second kind of order $\nu>0$  (see \cite{AShandbook}). 
\end{claim}
\begin{proof} Let $g_\sigma(t)=\bigl(\frac{\sigma^2}{2} \bigr)^{\nu_t/2-1}$ and $\mu$ have density $\frac{\sigma^2}{2}e^{-\sigma^2/2}$ with respect to Lebesque measure. Then equation (\ref{conveq}) becomes
\begin{align*} 
\int_0^\infty g_\sigma(s)g_\sigma(t) K_\sigma(s,t) d\mu(\sigma) &=  \det(\alpha_{st}^{-1/2}) \int_0^\infty  \left(\frac{\sigma^2}{2} \right)^{\nu_{st}-1}  \exp\left[-\frac{|\alpha_{st}^{-1/2}(s-t)|^2}{4(\sigma^2/2)}-\frac{\sigma^2}{2}\right] d\sigma \\
&\overset{x=\sigma^2/2}=\det(\alpha_{st}^{-1/2}) \int_0^\infty  x^{\nu_{st}-1}  \exp\left[-\frac{|\alpha_{st}^{-1/2}(s-t)|^2}{4x}-x\right] dx \\
&=2\det(\alpha_{st}^{-1/2}) 2^{-\nu_{s}/2} 2^{-\nu_{t}/2}|\alpha_{st}^{-1/2}(s-t)|^{\nu_{st}} \mathcal K_{\nu_{st}} (|\alpha_{st}^{-1/2}(s-t)|)
\end{align*}
where the last line is from (3.472.9) of Gradshteyn and Ryzhik \cite{Grad:book}.
Therefore $\det(\alpha_{st}^{-1/2}) \mathcal M_{\nu_{st}}\bigl(|\alpha_{st}^{-1/2}(s-t)|\bigr)$ is positive definite.
\end{proof}
 \bibliography{references}

\begin{thebibliography}{10}

\bibitem{AShandbook}
M.~Abramowitz and I.~Stegun.
\newblock {\em Handbook of Mathematical Functions}.
\newblock ninth ed. Dover, New York, 1965.

\bibitem{cholaug:02}
B.~Alkire.
\newblock Cholesky factorization of augmented positive definite matrices.
\newblock {\em Electrical Engineering Department, UCLA}, 2002.

\bibitem{1259022}
A.~Ayache, S.~Cohen, and J.~L. Vehel.
\newblock The covariance structure of multifractional brownian motion, with
  application to long range dependence.
\newblock In {\em ICASSP '00: Proceedings of the Acoustics, Speech, and Signal
  Processing, 2000. on IEEE International Conference}, pages 3810--3813,
  Washington, DC, USA, 2000. IEEE Computer Society.

\bibitem{Dahlh:00}
R.~Dahlhaus.
\newblock A likelihood approximation for locally stationary processes.
\newblock {\em Ann. Stat.}, 28:1762--1794, 2000.

\bibitem{fan:98}
J.~Fan, M.~Farmen, and I.~Gijbels.
\newblock Local maximum likelihood estiamtion and inference.
\newblock {\em J. Roy. Stat. Soc. B}, pages 591--608, 1998.

\bibitem{Grad:book}
I.~S. Gradshteyn and I.~M. Ryzhik.
\newblock {\em Table of Integrals, Series, and Products}.
\newblock Academic Press, San Diego, 2000.

\bibitem{muller:1987}
H.~G. M\"{u}ller.
\newblock Weighted local regression and kernel methods for nonparametric curve
  fitting.
\newblock {\em J. Amer. Stat. Assoc.}, 82:231--238, 1987.

\bibitem{paciorek:thesis}
C.~Paciorek.
\newblock {\em Nonstationary Gaussian processes for regression and spatial
  modelling}.
\newblock PhD thesis, Carnegie Mellon University, Department of Statistics,
  2003.

\bibitem{paciorek:04a}
C.~Paciorek and M.~Schervish.
\newblock Nonstationary covariance functions for gaussian process regression.
\newblock {\em Advances in Neural Information Processing Systems 16}, pages
  273--280, 2004.

\bibitem{paciorek:04b}
C.~Paciorek and M.~Schervish.
\newblock Spatial smoothing using a new class of nonstationary covariance
  functions, available at {\tt www.biostat.harvard.edu/\~{}paciorek}.
\newblock 2004.

\bibitem{peltier:95}
R.~Peltier and J.~L\'evy V\'ehel.
\newblock Multifractional brownian motion: definition and preliminary results.
\newblock Technical report, Res. Rept. 2645, INRIA, 2005.

\bibitem{pintore:04}
A.~Pintore and C.~Holmes.
\newblock Spatially adaptive non-stationary covariance functions via spatially
  adaptive spectra, available at {\tt
  www.stats.ox.ac.uk/\~{}cholmes/reports/spectral\_tempering.pdf}.
\newblock 2004.

\bibitem{sly:07}
A.~Sly.
\newblock Integrated fractional white noise as an alternative to
  multifractional brownian motion.
\newblock {\em J. Appl. Prob.}, pages 393--408, 2007.

\bibitem{stein:book}
M.~L. Stein.
\newblock {\em Interpolation of Spatial Data: Some Theory for Kriging}.
\newblock Springer, New York, 1999.

\bibitem{stein:nonstatTech}
M.~L. Stein.
\newblock Nonstationary spatial covariance fucntions.
\newblock Technical report, University of Chicago, Department of Statistics,
  2005.

\bibitem{stoev:06}
S.~Stoev and M.~Taqqu.
\newblock How rich is the class of multifractional brownian motions?
\newblock {\em Stoc. Proc. App.}, pages 200--221, 2006.

\bibitem{wand:90}
M.~P. Wand and W.~R. Schucany.
\newblock Gaussian-based kernels.
\newblock {\em The Canadian Journal of Statistics}, 18(3):197--204, 1990.

\end{thebibliography}
 \end{document}